\documentclass[12pt]{article}

\usepackage[utf8]{inputenc}
\usepackage[ruled, vlined]{algorithm2e}
\usepackage{amsfonts, amsmath, amssymb, amsthm, bbm, color, enumerate, graphicx, mathtools, tikz, hyperref, relsize, bm}
\usepackage[style = authoryear, maxcitenames = 2, natbib]{biblatex}
\usepackage[margin=1.5in,footskip=0.25in]{geometry}
\usepackage[affil-it]{authblk}

\newtheorem{theorem}{\bf Theorem}

\newtheorem{proposition}{Proposition}

\newtheorem{otherassumption}{Assumption}

\makeatletter
\renewcommand{\theotherassumption}{A\@arabic\c@otherassumption}
\makeatother

\newcommand{\n}{\nu}

\newcommand{\s}{\sigma}
\newcommand{\G}{\Gamma}
\newcommand{\rh}{\rho}

\title{Constrained Fiducial Inference for Gaussian Models}
\author[1]{Hank Flury}
\author[1]{Jan Hannig}
\author[1]{Richard Smith}
\affil[1]{University of North Carolina at Chapel Hill, Department of Statistics and Operations Research, Chapel Hill, North Carolina, United States}

\begin{document}
\maketitle
Hank Flury is the corresponding author. His email is: hankflury@gmail.com. Hank Flury's ORCiD is 0009-0004-1722-8607. Richard Smith's ORCiD is 0000-0002-3454-3199. Jan Hannig's ORCiD is 0000-0002-4164-0173.

Jan Hannig's research was supported in part by the National Science Foundation under Grant Nos. DMS-1916115, 2113404, and 2210337. Hank Flury was also supported in part by the National Science Foundation under Grant No. DMS-2210337.

\newpage

\begin{abstract}
     We propose a new fiducial Markov Chain Monte Carlo (MCMC) method for fitting parametric Gaussian models. We utilize the Cayley transform to decompose the parametric covariance matrix, which in turn allows us to formulate a general data generating algorithm for Gaussian data. Leveraging constrained generalized fiducial inference, we are able to create the basis of an MCMC algorithm, which can be specified to parametric models with minimal effort. The appeal of this novel approach is the wide class of models which it permits, ease of implementation and the posterior-like fiducial distribution without the need for a prior. We provide background information for the derivation of the relevant fiducial quantities, and a proof that the proposed MCMC algorithm targets the correct fiducial distribution. We need not assume independence nor identical distribution of the data, which makes the method attractive for application to time series and spatial data. Well-performing simulation results of the MA(1) and Mat\'ern models are presented. 
\end{abstract}

\section{Introduction}
Gaussian models are commonly utilized in spatial and time series data analysis, often because of their flexibility. As such, rich theory for spatial and temporal data has been developed around Gaussian models; see for example \citet{brockwellDavis} and \citet{stein99}. Consider the fitting of data to a parameterized Gaussian model from which the mean and covariance function are easily extracted. As a simple example of such a parameterization, we consider the MA(1) model. The MA(1) model is generated by the relationship
\begin{equation}
    X_{n} = \epsilon_n+\rho\epsilon_{n-1}, \hspace{5mm} \epsilon_n \sim N(0,\sigma^2) \hspace{1mm}\forall n \in \mathbb{Z}.
\end{equation}
This relationship produces the covariance matrix: 
\begin{equation}\label{MA(1)Matrix}
    \sigma^2\begin{bmatrix} 
    1 + \rho^2 & \rho & 0 & 0 & \dots \\
    \rho & 1 + \rho^2 & \rho & 0 & \\
    0 & \rho & 1 + \rho^2 & \rho & \\
    0 & 0 & \rho & 1 + \rho^2 & \\
    \vdots &  &  &  & \ddots
    \end{bmatrix}.
\end{equation}
To ensure identifiability, restrict $|\rho| \leq 1$ \citet{brockwellDavis}. The MA(1) model is relatively simple. For a more complicated model, consider the Mat\'ern model. The covariance function of the Mat\'ern model is given by
\begin{equation}
    C_\nu(d) = \s^2\frac{2^{1-\n}}{\G(\n)}\left(\sqrt{2\n}\frac{d}{\rh}\right)^{\n}K_{\n}\left(\sqrt{2\n}\frac{d}{\rh}\right),
\end{equation}
where $\G$ is the gamma function, $K_{\nu}$ is the modified Bessel function of the second kind, and the input $d$ is the distance between sites. This model is common in spatial applications, but the relative complexity of its form may make it more difficult to fit. In fact, the Mat\'ern model is particularly challenging to fit using Bayesian methods, as certain priors on its parameters result in non-integrable posteriors \citep{Berger01122001}. The fiducial model which we present in this work both circumvents the need to select a prior, and results in a posterior-like distribution which is guaranteed to be a proper probability distribution \citep{Hannig02072016}.

Fitting these models is usually done through maximum likelihood estimation or via the Bayes estimator. This work considers an alternative method in fiducial inference. In addition to the usual attractive qualities of fiducial inference, namely a Bayesian posterior-like distribution without the need to specify a prior, the advantage of our approach is its ease of implementation. We provide the basis for a general (MCMC) algorithm which allows for the fitting of Gaussian models. This method requires a practitioner to provide little more than a function which returns the covariance matrix and its derivative with respect to a parameter vector to fit any particular Gaussian model. Moreover, this approach permits a wide class of models. Specifically, we need not assume that our data are i.i.d., which makes this method especially attractive for application to time series data. We implement this algorithm in Matlab, the code for which can be found at \url{https://github.com/fluryh/Constrained-Fiducial-Inference-for-Gaussian-Models}.

This work builds on that of \citet{murph2022generalized}, which provides the basis for fiducial inference constrained to a manifold. The approach in this paper differs from that of \citet{murph2022generalized} in a few ways. First, the previous paper develops an approach for constrained fiducial inference in situations for which a level set is a natural parameterization for the constraining manifold. By contrast, in the motivating examples of our work, the natural parameterization is the form of the covariance matrix itself. As such, we work to create a method for which beginning at a parameterized covariance matrix is natural. Finally, both this work and that of \citet{murph2022generalized} utilize the Cayley transform, a decomposition of an orthogonal matrix which is unique up to multiplication by a signature matrix $Z$. In previous work a single signature matrix was selected at each iteration of the algorithm. In our work, we sample from all potential signature matrices and average over the sample, in an effort to reduce the influence of the arbitrary selection of a signature matrix. Another related work, \citet{liang2025extended}, proposes automatic fiducial inference based on large neural network. 

Our work adds to the growing base of literature on generalized fiducial inference. \citet{Du02102025} introduces AutoGFI, a method which estimates model parameters with the need for the user to specify a fiducial distribution. This is similar to our work in that it relieves the user from the specifics of determining an appropriate fiducial distribution, though for entirely different models and settings. The recent work of \citet{Yang2025} also considers the geometric perspective of fiducial inference, including usage of the Cayley transform and projection onto manifolds. Our work is closely related to these results, though we consider the unique setting of time series and spatial data. Alternatively, \citet{Williams2023} develop a fiducial method for multivariate time series data, specifically focusing on graphical vector autoregression models. Though there is some overlap in the data being considered, their approach is entirely distinct from ours, as they develop an epsilon-admissible subset procedure. Other recent advancements in GFI include \citet{Du2025}, which develops GFI procedures for high dimensional linear regression on sparse data, \citet{Cui02072024} which focuses on a GFI approach for interval-censored survival data, and \citet{Yi2022}, which leverages fiducial inference's posterior-like distribution to provide uncertainty quantification for spectral lines in signal processing. 

The rest of this work is structured as follows. Section \ref{Fiducial derivations} details the approach taken for the inference of Gaussian parametric models and provides derivations of needed quantities. Section \ref{pseudo} presents pseudocode for the implementation of our proposed MCMC algorithm. Section \ref{Fiducial simulations} describes the results of our algorithm on simulated data from the MA(1) and Mat\'ern models. Finally, Section \ref{Fiducial proofs} contains a proof that the proposed MCMC algorithm targets the correct fiducial distribution, and cites sufficient conditions for the asymptotic normality and consistency of the estimator.

\section{Derivation of the Fiducial Distribution}\label{Fiducial derivations}
We begin with the parametric representation of the data generating algorithm (DGA). Let $\theta \in \mathbb{R}^p$ be the vector which parameterizes the Gaussian model of interest. Let $\Sigma(\theta): \mathbb{R}^p \to \mathbb{R}^{n\times n}$ be the function which defines the covariance matrix, and $\mu(\theta): \mathbb{R}^p \to \mathbb{R}^n$ be the function which defines the parametric mean vector. We make the following assumptions on these parameterizations:
\begin{enumerate}
    \item $\theta \in \Theta \subset\mathbb{R}^p$, where $\Theta$ is open. The vector $\theta$ represents the minimum number of parameters required to parameterize the given model.
    \item $\Sigma(\theta)$ is positive definite and has unique eigenvalues for all $\theta \in \Theta$.
    \item The function $\theta \to (\mu(\theta),\Sigma(\theta))$ is continuous and injective.
    \item Both $\mu(\theta)$ and $\Sigma(\theta)$ are twice differentiable on $\Theta$.
\end{enumerate}These assumptions imply that $\theta\to (\mu(\theta),\Sigma(\theta))$ defines a $p$ dimensional manifold. We note that the mean and variance functions $\mu$ and $\Sigma$, may include covariates (as is the case for the Mat\'ern covariance model, which requires the distance between sites), but we consider this to be part of the functions themselves.

For Gaussian data, the definitions of $\mu$ and $\Sigma$ immediately allows us to define a DGA. First, we find the function $\Sigma^{1/2}(\theta)$ such that $\Sigma^{1/2}(\theta)\Sigma^{1/2}(\theta) = \Sigma(\theta)$. Let $U$ be a vector of independent standard normal random variables. Then our DGA may be written
\begin{equation}\label{parametricDGA}
     Y(U,\theta) := \mu(\theta) + \Sigma^{1/2}(\theta)U.
\end{equation}From this, we can take in some randomness in the form of $U$, and observe the data $Y$, which will follow a normal distribution with variance $\Sigma(\theta)$, and mean $\mu(\theta)$. From here, we could define the generalized fiducial distribution (GFD) as
\begin{equation}\label{GFD}
     \frac{f(y|\theta)D\left(\nabla_{\theta}Y(u,\theta)|_{u=Y^{-1}(y,\theta)}\right)}{\int_{\Theta}f(y|\theta')D\left(\nabla_{\theta'}Y(u,\theta')|_{u=Y^{-1}(y,\theta')}\right)d\theta'}
\end{equation} where $D(X) = \det(n^{-1}X^TX)^{-1/2}$ \citep{Hannig02072016}. The matrix $\nabla_{\theta}Y(u,\theta)|_{u=Y^{-1}(y,\theta)}$ is the gradient of the data generating algorithm evaluated at values of $u$ which lead to the observed values $Y$. We note that this $u$ need not be unique; for details of this argument, see \citet{Hannig02072016}. To calculate this quantity in practice, we would first find $\nabla_\theta Y(u,\theta)$. Then we would invert the data generating algorithm to solve for $U$. In our example this gives us
\begin{equation*}
    U = \Sigma^{-1/2}(\theta)\left(Y-\mu(\theta)\right).
\end{equation*} That is, we could write
\begin{equation*}
    \nabla_{\theta}Y(u,\theta)|_{u=Y^{-1}(y,\theta)} = \nabla_\theta Y(\Sigma^{-1/2}(\theta)\left(Y-\mu(\theta)\right),\theta).
\end{equation*}

While the above may work for theoretical uses, in practice the form of $\Sigma^{1/2}(\theta)$, and therefore $\Sigma^{-1/2}(\theta)$ and $\nabla_\theta\Sigma^{1/2}(\theta)$, may be prohibitively difficult to find. Without $\Sigma^{ 1/2}(\theta)$, we cannot write down the GFD, let alone evaluate it. As such, we need another method by which we can describe the data generating algorithm. Following \citet{murph2023introduction}, we utilize the Cayley Transform to write our DGA. We begin with our covariance matrix, $\Sigma(\theta)$. This matrix is assumed to be positive-definite, and therefore we may take the compact SVD, which we denote $\Sigma = S\Lambda^2S^T$, where $\Lambda$ is a diagonal matrix of positive elements and $S$ is an orthogonal matrix. We reparameterize $S$ by making use of the Cayley transform as described in Theorem \ref{cayley} of \citet{murph2022generalized}. We cite the following two theorems regarding the Cayley transform:
\begin{theorem}[Cayley Transform, \citet{eves}] \label{firstcayley}
    Every real orthogonal matrix S, that does not have characteristic root -1 may be represented as
    \begin{equation}
        S = (I_d - A)(I_d + A)^{-1}
    \end{equation}
    where $I_d$ is the identity matrix of the same dimension as $S$, and $A$ is a skew-symmetric matrix ($A^T=-A)$. 
\end{theorem}
\begin{theorem}[O'Dorney, \citet{odorney2014}] \label{cayley}
    For every real orthogonal matrix S, there exists a signature matrix (a matrix with $\pm1$ on the diagonal, and zero everywhere else) Z such that
    \begin{equation*}
        SZ = (I_d-A)(I_d + A)^{-1}
    \end{equation*} for a skew symmetric matrix $A$, and the elements of $A$ are between $-1$ and $1$, inclusive. 
\end{theorem}Note the requirement that $S$ must not have $-1$ as a characteristic root. This could cause a problem, as $S$ very well could have $-1$ as a characteristic root. However, we also note that multiplication of $S$ by some signature matrix $Z$ does not affect our distribution. Namely, if $S$ is the result of the SVD such that $S\Lambda^2S = \Sigma$, then so too does $SZ\Lambda^2ZS$. Theorem \ref{cayley} ensures that there must exist a signature matrix $Z$ such that $SZ$ does not have $-1$ as a characteristic root. So, for any $\theta$, we can find at least one pair $A,\Lambda$ such that we can define $\Sigma(\theta)$ by a diagonal matrix $\Lambda$ and a skew-symmetric matrix $A$. Using this Cayley transform, and rewriting the mean solely as a vector $\mu \in \mathbb{R}^n$, we can rewrite the DGA as
\begin{equation}\label{Cayley DGA}
    Y  = Y(A,\Lambda,\mu,U) := Y(M,U) := \mu + (I-A)(I + A)^{-1}\Lambda U,
\end{equation}where $M$ represents the tuple, $(A,\Lambda,\mu)$. The Cayley transform then allows us to fit the data, and ensure that the values fit result in a valid covariance matrix. However, we note that this DGA contains no information about the specific family of models which we would like to fit, and that it is overspecified: for a vector of $n$ samples, we have $n + \frac{n(n+1)}{2}$ parameters. 

We now have two DGAs, each with their own advantages. The DGA formulated in Equation (\ref{parametricDGA}) specifies the model of interest, but may be difficult to write in an explicit form. The DGA formulated in Equation (\ref{Cayley DGA}) can be easily written and is completely general, but is not restricted only to the models of interest. As such, our aim is to utilize the DGA of Equation (\ref{Cayley DGA}), but constrain it only to covariance matrix and mean vectors permitted by the model. This is where we may use Constrained Generalized Fiducial Inference results of \citet{murph2022generalized}. Starting from the GFD given in Equation (\ref{GFD}), we can apply a change of variable to transfer from the $\theta$ space to the $(A,\Lambda)$ space. This results in the distribution
\begin{equation}\label{first fiducial dist}
    r(M|y) =\frac{f(y|M)D^*(\nabla_{M} Y(u,M)|_{u = Y^{-1}(y,M)}P_{M})}{\int_\mathcal{M} f(y|M')D^*(\nabla_{M'} Y(u,M')|_{u = Y^{-1}(y,M)}P_{M'})dM'}, \quad M \in \mathcal{M},
\end{equation}
where $\mathcal{M}$ is the manifold within the parameter space of $M$ corresponding to values which keep the mean vector and covariance matrix in the family allowed by the parameterization $\theta$. We call this distribution the Generalized Constrained Fiducial Distribution (GCFD). $P_{M}$ is the projection matrix which projects orthogonally onto the manifold $\mathcal{M}$, thereby constraining the values of the mean vector and covariance matrix to only those allowed by the model. Following \citet{murph2022generalized}, we note that $D^*(\nabla_{M} Y(z,M)|_{z = Y^{-1}(y,M)}P_{M}) = D(\nabla_{M} Y(z,M)|_{z = Y^{-1}(y,M)}Q_{M})$ where $Q_M$ is the compact SVD of $P_M$ such that $Q_M Q_M^T = P_M$.

We return now to the complication that in order to use the Cayley transform, $-1$ cannot be a characteristic root of $S$. First, apply some arbitrary order to the set of signature matrices, whose determinant is $1$, so that we may specify them as $Z_i$, $i \in 1,...,2^{d-1}$. Although there are $2^d$ total possible signature matrices, we will only need to consider those with determinant 1, thus we may reduce the ordering to only $2^{d-1}$ elements. This detail is further discussed in Section \ref{pseudo}. Define $\Theta_i$ as the subset of $\Theta$ such that for all $\theta \in \Theta_i$, $S(\theta)Z_i$ does not have $-1$ as a characteristic root. In order to continue to use the basis provided by \citet{murph2022generalized}, we must split $\Theta$ into each $\Theta_i$ so that we may still smoothly transition from the $\theta$ space to the $(A,\Lambda)$ space. Since the $\Theta_i$ are not disjoint, this means we target the distribution
\begin{equation}\label{summed}
    \frac{\sum_{i=1}^{2^{d-1}} f(y|M_i)D^*(\nabla_{M} Y(u,M_i)|_{u = Y^{-1}(y,M_i)}P_{M_i})1_{\mathcal{M_i}}(M_i)}{\sum_{i=1}^{2^{d-1}}\int_{\mathcal{M}_i} f(y|M_i')D^*(\nabla_{M} Y(u,M_i')|_{u = Y^{-1}(y,M_i)}P_{M_i'})dM_i'}, \quad M \in \mathcal{M},
\end{equation} where $\mathcal{M}_i$ corresponds to $\Theta_i$.

In order to utilize the above distribution, we need to find the form of three quantities: $\nabla_{M}Y|_{u = Y^{-1}(y,M)}$, $P_{M}$ and the change of variable Jacobian. The first of these quantities, $\nabla_{M}Y|_{u = Y^{-1}(y,M)}$ is straightforward to calculate. Given the DGA as written in Equation (\ref{Cayley DGA}), we may write

\begin{equation*}
    \frac{\partial}{\partial \Lambda_{i,i}} Y(M,U) = (I-A)(I + A)^{-1}J^{i,i}U
\end{equation*}
and, 
\begin{align*}
    \frac{\partial}{\partial A_{i,j}} Y(M,U) &= \left(\frac{\partial}{\partial A_{i,j}}\left[I-A\right](I-A)^{-1} + (I-A)\frac{\partial}{\partial A_{i,j}}[(I+A)^{-1}]\right)\Lambda U\\
    &= -2(I+A)^{-1}(J^{i,j} - J^{j,i})(I + A)^{-1}\Lambda U,
\end{align*}
where $X_{i,j}$ is the element of a matrix $X$ at row $i$, column $j$ and $J^{i,j}$ is a matrix with all entries equal to zero, except for a 1 at row $i$, column $j$. We have
\begin{equation*}
    \frac{\partial}{\partial \mu_i} Y(M,U) = K_i
\end{equation*}
where $K_i$ is a vector of zeros, except for a $1$ at the $i$th position. The above formulas give us $\nabla_{M}Y$. To restrict the value of U we solve the DGA for $U$, leading to
\begin{equation*}
    U = \Lambda^{-1}(I+A)^{-1}(I-A)(y-\mu)
\end{equation*}
where $y$ are our data. Finally, we have
\begin{equation*}
    \frac{\partial}{\partial \Lambda_{s,s}} Y|_{u = Y^{-1}(y,M)} = (I-A)(I + A)^{-1}J^{s,s}\Lambda^{-1}(I+A)(I-A)^{-1}(y-\mu)
\end{equation*}
and 
\begin{equation*}
     \frac{\partial}{\partial A_{i,j}} Y|_{u = Y^{-1}(y,M)} = -2(I+A)^{-1}(J^{i,j} - J^{j,i})(I-A)^{-1}(y-\mu)
\end{equation*} and the derivative with respect to the $\mu$ parameters remaining unchanged. With these equations, we are able to build $\nabla_{M}Y|_{u = Y^{-1}(y,M)}$.

The next step to calculating the GCFD is calculating $P_{M}$. Let $G: \mathbb{R}^p \rightarrow \mathbb{R}^\frac{d(d+1)}{2}$ be the function which takes in our parameter vector $\theta$ and returns the upper triangular entries (that is, the unique entries) of the covariance matrix $\Sigma$. Similarly, define $H:\mathbb{R}^\frac{d(d+1)}{2} \rightarrow \mathbb{R}^\frac{d(d+1)}{2}$ to be the function which takes in the unique parameters of $A$ and $\Lambda$, that is, the $\frac{d(d-1)}{2}$ elements of A above the diagonal and the $d$ diagonal elements of $\Lambda$, and returns $\Sigma$. Note then, that $H$ will be the same, regardless of the specific model being used. Then, we have the relationship:
\begin{equation}\label{parameterizations}
    \theta \xrightarrow{G} \Sigma \xleftarrow{H} M.
\end{equation}
Note, however, that we cannot directly invert $H$, as we may potentially choose from a number of signature matrices $Z_i$. Thus, we write the function $H_i$ as the $H$ function associated with the signature matrix $Z_i$, so that the function $H_i^{-1}(G(\theta))$ defines our manifold in terms of $M$. From this, we have that the form for our projection matrix is 
\begin{equation*}
    P_{M} = (\nabla_{M}H_i)^{-1}\nabla_\theta G( \nabla_\theta G^T(\nabla_{M}H_i)^{-T} (\nabla_{M}H_i)^{-1}\nabla_\theta G)^{-1} \nabla_\theta G^T(\nabla_{M}H_i)^{-T}.
\end{equation*}
While in theory this is sufficient to write our GCFD, it will later be advantageous to formulate $Q_{M}$, such that $Q_MQ_M^T=P_M$. To do so, we define a $\frac{d(d+1)}{2} \times p$ matrix $V$ such that $VV^T = ( \nabla_\theta G^T(\nabla_{M}H_i)^{-T} (\nabla_{M}H_i)^{-1}\nabla_\theta G)^{-1}$. Note that the assumption that $p$ is the minimum number of parameters needed to parameterize covariance matrix assures that we can find such a matrix. For the moment, we need not find any closed form, just define this relationship. As such, we have $Q_{M} = (\nabla_{M}H)^{-1}\nabla_\theta GV$.

Finally, we need to define a change of variable Jacobian. We use the relationships defined in (\ref{parameterizations}), with the only complication that we need to account for the difference in dimensions between the parameterizations. Namely, since our manifold is actually $p$ dimensional, we need to project $\nabla_{M}H$ down to $p$ dimensions. Just as before, we could formulate this as $D^*(\nabla_{M}HP_{M})$ but instead we use $D(\nabla_{M}HQ_{M})$. This gives us the determinant of the change of variable Jacobian:
\begin{equation*}
\det(Q_{M}^T\nabla_{M}H^T\nabla_{M}HQ_{M})^{-1/2}\det(\nabla_\theta G^T\nabla_\theta G)^{1/2}.
\end{equation*}
With this, we can finally write the GCFD in terms of $M$:
\begin{align*}
    r(M|y) &\propto f(y|M)D(\nabla_{M}YQ_{M})D(\nabla_{M}HQ_{M})^{-1}D(\nabla_\theta G)\\
    &= f(y|M)D(\nabla_{M}Y(\nabla_{M}H)^{-1}\nabla_\theta G),
\end{align*}
where the second equality comes from the fact that $\nabla_\theta GM$ is square.

\section{Pseudocode}\label{pseudo}
In this section, we make explicit a few details of the algorithm and introduce some notation. Recall that the decomposition of $\Sigma$ is unique only up to multiplication of $S$ by some signature matrix $Z$, and that the Cayley transform can only be applied to orthogonal matrices which do not have -1 as an eigenvalue. As such, not all signature matrices can be used for the decomposition at each value $\theta$. The signature matrices have a non-trivial effect on the unnormalized distributions. The majority of the details of the algorithm are concerned with accounting for the signature matrices.

The idea behind this algorithm is simple: we propose a new value for $\theta$, and evaluate it after being parameterized by $(A,\Lambda)$. Since this parameterization by $(A,\Lambda)$ is only unique up to a signature matrix, we evaluate the proposal at a number of these signature matrices and average out their effects. We denote the number of sampled signature matrices by $k$. In higher dimensions, there are so many possible signature matrices that even with a reasonable $k$, if signature matrices were chosen uniformly at random, a proposal might have no signature matrices in common with the current position in the parameter space. In that case, it could be the selected signature matrices, not the parameters and data, which drives the acceptance or rejection step of the MCMC algorithm. To deal with this, we randomize only a proportion of the signature matrices to force some overlap in the signature matrices. We denote the number of signature matrices we keep by $q$. The signature matrices that are kept are chosen uniformly at random, and the new matrices are also chosen uniformly at random from all available signature matrices, with replacement. It is possible that of the $k$ sampled matrices, none are permissible by the proposal $\theta'$. In this case, we consider the algorithm to have left the parameter space and immediately reject the proposal. 

We also note that we do not need to sample all possible matrices. Recall that the matrix $S$ is orthonormal, and thus has determinant $\pm1$, and its eigenvalues lie on the complex unit circle. The same is true of $SZ_i$. Additionally, the existence of a strictly complex eigenvalue implies that the conjugate of the eigenvalue must also itself be an eigenvalue. As such, if $SZ_i$ has determinant $-1$, then $-1$ must be a characteristic root of the matrix. This means that we only need to sample from $2^{d-1}$ possible matrices, namely those with the same determinant as $S$. To address the issue that $Z_i$ will only work for $S$ with a determinant of one sign, when we begin a new proposal if $S$ has a negative determinant, we multiply the quantity by the signature matrix with $-1$ in the first diagonal position and $1$ in all the rest. We call this signature matrix $Z_1$. That way, we only consider signature matrices with positive determinant. With these details clarified, we present the pseudocode for a single step in the MCMC algorithm.

\begin{algorithm}
\caption{Single Step of Fiducial MCMC Pseudocode}\label{alg:cap}
\SetKwInOut{Input}{input}\SetKwInOut{Output}{output}
\Input{$\theta_n$, $\{Z_i\}$, $k$, $p<k$, pl = previous likelihood, ps = previous sum}
Generate new proposal, $\theta'$\, from normal distribution centered at $\theta_n$\;
\eIf{$\theta$ falls outside the parameter space}{
    $\theta_{n+1} = \theta_n$\;
    }{
    $\Sigma$ = covariance matrix at $\theta'$\;
    $[S,D]$ = svd($\Sigma$)\;
    \If{$\det(S) == -1$}{
    $S = SZ_1$
    }
    l = log-likelihood function at proposal\;
    Sum $= 0$\;
    $\{Z_i^*\} = p$ uniformly selected $Z_i$ from $\{Z_i\}$\;
    $\{\tilde{Z}_i\} = k-q$ uniformly selected $Z_i$ from all $2^{d-1}$ possible $Z_i$\;
    $\{Z_i'\} = $ concatenation of $\{Z_i^*\}$ and $\{\tilde{Z}_i\}$\;
    \For{$Z_i' \in \{Z_i'\}$}{
        \If{$SZ_i'$ does not have $-1$ as a characteristic root}{
        $A_i = (I - SZ_i')/(I + SZ_i')$\;
        $P$ = projection matrix in $(A_i,\Lambda)$ space\;
        $Q$ = svd($P$)\;
        Sum = Sum + Determinant Terms of GFCD\;
        }
    }
    \eIf{Sum $== 0$}{
    $\theta_{n+1} = \theta_n\:$
    }{MHR = Sum*l/(pl*ps)\;
    \eIf{MHR $\geq$ \text{Uniform[0,1]}}{
    $\theta_{n+1} = \theta'$\;
    $\{Z_i\} = \{Z_i'\}$\;
    }{
    {$\theta_{n+1} = \theta_n$\;}
    }}
    }
\end{algorithm}

The advantage of this algorithm is its ease of use and its application to a large family of models. The user need only provide three functions: $G(\theta)$, $\nabla_\theta G(\theta)$ and a function which checks that a given $\theta$ is a valid set of parameters for the given family (i.e. a correlation parameter is between -1 and 1, or a variance parameter is non-negative). In its current implementation, the step size for the proposal is generated by a centered Gaussian distribution, though this is not strictly necessary. It is recommended that the user specify what variance should be used for this proposal step. For example, a correlation parameter which takes values between -1 and 1 may need smaller steps than a variance parameter which may take any positive value.

\section{Simulations}\label{Fiducial simulations}
To support the analysis of our algorithm, we run several simulations using the aforementioned MA(1) and Mat\'ern covariance matrices. The first model we consider is the MA(1) model. The covariance matrix is parameterized by $\theta = (\rho, \sigma^2)$. To ensure that the model is identifiable, the parameters are restricted to $-1 \leq \rho \leq 1$ and $0 < \sigma^2$. We generated 20 independent observations of a length 50 MA(1) vector. The data were generated with parameter values of $\sigma^2 = 6$ and $\rho = .5$. The covariance matrix takes the form

\begin{equation}
    \begin{bmatrix} 
    7.5 & 3 & 0 & 0 &\dots & 0 &0 \\
    3 & 7.5 & 3 & 0 &\dots & 0 &0 \\
    0 & 3 & 7.5 & 0 &\dots & 0 &0\\
    0 & 0 & 3 & 7.5 & \dots & 0 & 0\\
    
    \vdots & \vdots & \vdots & \vdots &\ddots & \vdots & \vdots\\
    0 & 0 & 0 &  0 & \dots & 7.5 & 3\\
    0 & 0 & 0 &  0 & \dots & 3 & 7.5
    \end{bmatrix}.
\end{equation}

Other than a function which checks that the parameters are within the parameter space, we provide only two functions unique to the MA(1) process: a function which, given $(\rho,\sigma^2)$, returns the covariance matrix, and a function that returns $\nabla G(\rho,\sigma^2)$, or more precisely the first derivative of the upper triangular elements of the covariance matrix taken with respect to both parameters. It is easy to show that this is:

\begin{equation}
    \begin{bmatrix} 
     1+\rho^2 & \rho & 0 & \dots & 1+\rho^2 & \rho & 0 & 1+\rho^2 & \rho & 1+\rho^2\\
     2\sigma^2\rho & \sigma^2 & 0 & \dots & 2\sigma^2\rho & \sigma^2 & 0 & 2\sigma^2\rho & \sigma^2 &  2\sigma^2\rho
    \end{bmatrix}^T.
\end{equation}

We initialize the algorithm at the value of $\theta = (.8, 2)$ far from the true values used to generate the data. This is done to demonstrate the algorithm's ability to stabilize around the correct values. Although a practitioner would usually choose to initialize the algorithm at an MLE or some other well-performing estimator, the performance of the algorithm in this case supports the conclusion that the algorithm will converge to the correct values independent of the initialization. The  MCMC algorithm is run for 6000 total steps, considering the first 1000 to be initialization steps and excluding them from any calculations. We repeat the experiment 200 times to provide some estimate of the uncertainty. Figure \ref{fig:MA(1) sims} shows layered histograms which display both the MCMC estimates as well as the MLE for estimates for both parameters for each of the 200 runs. In Figure \ref{fig:MA(1) scatter}, we also provide a scatter plot which again compares the MLE estimates and MCMC estimates, such that we can observe that the joint distribution also appears normal. From these two figures, we provide evidence that the distribution of the MLE coincides with that of our MCMC estimates, supporting the idea that our method leads to accurate estimation and indeed that the conditions for the fiducial Bernstein-von Mises result in Section \ref{Fiducial proofs} are met. We also report an average acceptance rate of 29.45\%, with a minimum acceptance rate of 16.94\% and a maximum of 42.52\%. Table \ref{estimated vals} summarizes the results from this experiment.

One advantage of the fiducial approach is uncertainty quantification. That is, because the output of our algorithm is a Markov chain, we have a distribution as our estimator of the parameters, as opposed to a point estimate in the case of the MLE. We record the 2.5th and 97.5th quantiles of the chain for all runs. We see that for the $\rho$ parameter, this interval produces a coverage probability of 80.5\%, and 93\% for the $\sigma^2$ parameter. The joint coverage probability is 74.5\%. These values are slightly lower than expected, likely due to a combination of a relatively small sample size in the number of runs, and short MCMC chains. Table \ref{estimated vals} also reports the average 2.5th and 97.5th quantiles. 


\begin{figure}
\centering
\begin{minipage}{.5\textwidth}
\centering
     \includegraphics[width=\textwidth]{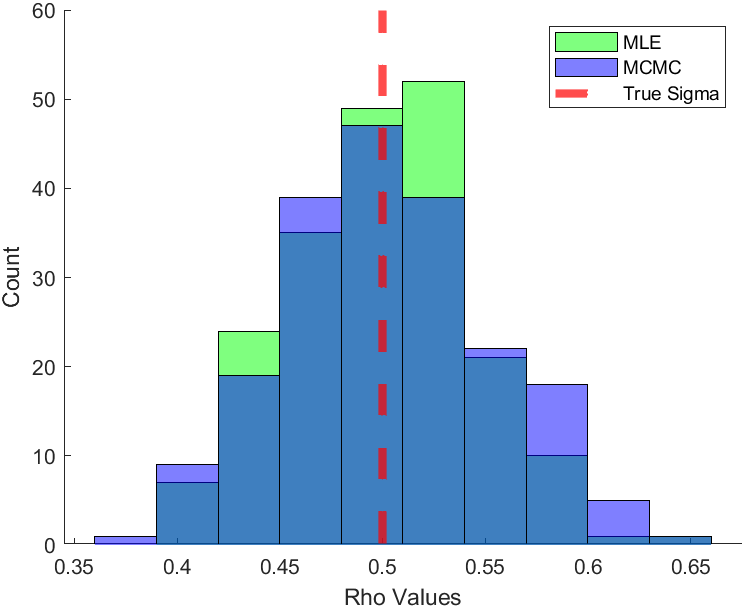}
\end{minipage}\hfill
\begin{minipage}{.5\textwidth}
    \centering
    \includegraphics[width=\textwidth]{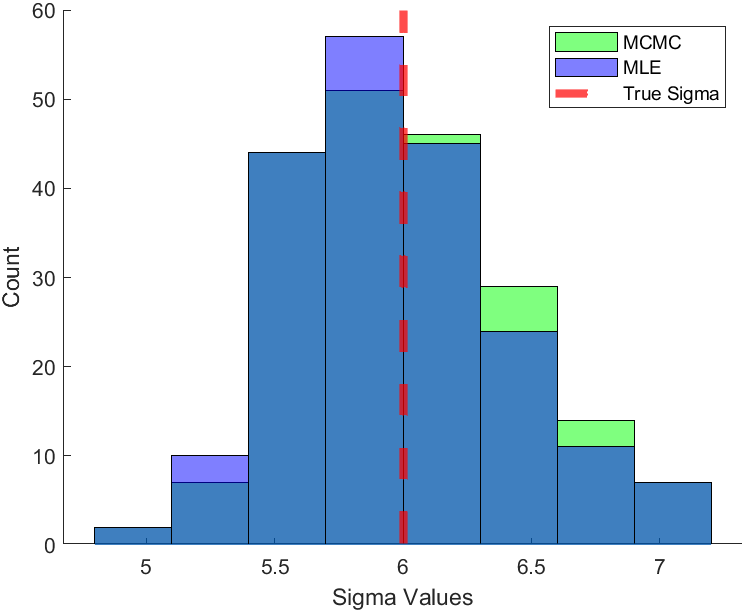}
\end{minipage}
\caption{Left: Estimate of $\rho = .5$ from 200 runs of our algorithm, compared to the MLE for the same data. Right: Estimate of $\sigma^2 = 6$ from 200 runs of our algorithm, compared to the MLE for the same data.}
\label{fig:MA(1) sims}
\end{figure}

\begin{figure}
    \centering
    \includegraphics[width=0.75\linewidth]{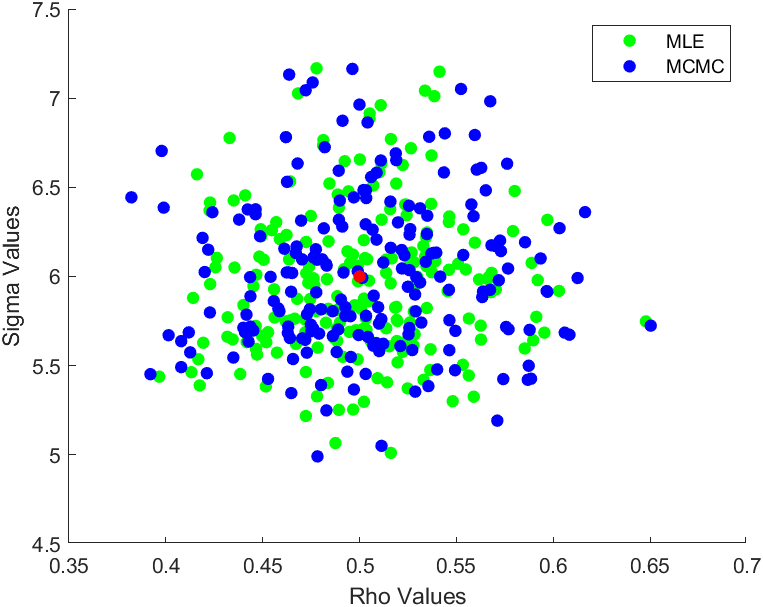}
    \caption{Scatter plot of estimated parameter values for each of the 200 runs on simulated MA(1) data. The red dot indicates the true parameter value.}
    \label{fig:MA(1) scatter}
\end{figure}


We also run another simulation on the more complicated Mat\'ern covariance matrix. We assume we observe 50 sites, arranged in a 5 by 10 unit grid with random jitter. Our data consist of 50 replications of this length 50 vector with parameters $(\nu,\sigma^2,\rho) = (2,6,1)$. It is well known that for the Mat\'ern model, the $\sigma^2$ and $\rho$ parameters are not separately identifiable, which causes correlated posteriors for MCMC approaches, see \citet{stein99} and \citet{Zhang2004}. To alleviate the effects of this on our algorithm, we amend the proposal step slightly, to take a step in one of the parameters, leaving the other two the same. Each step, it rotates between which parameter has a proposed step. We run the simulation for 5000 steps, starting from the true parameter value. Once again, we see that the estimates produced by our procedure have similar performance to the MLE. Figure \ref{fig:Matern sims} shows the  histograms for the Mat\'ern parameters, and Table \ref{estimated vals} reports the MLEs and MCMC estimates for both the MA(1) and Mat\'ern models. Figure \ref{fig:Matern scatter} also provides scatter plots for the MCMC and MLE estimates of the parameters on each run, for each pair of parameters. The average acceptance rate was 22.48\%, with a minimum of 19.02\% and a maximum of 25.40\%. We report a coverage probability by the 97.5th and 2.5th percentiles of 69.5\% for $\nu$, 91\% for $\sigma^2$, and 67\% for $\rho$, and a joint coverage of 47\%. We again note that these coverage probabilities would likely improve from and increased number of runs, as well as longer chains.

We address one of the potential weaknesses of this method: the computational load as array length increases. Since the method relies on overspecifying the mean and covariance matrices, this means the number of overspecifying parameters scales quadratically with array length. Further, the algorithm requires inversion and multiplication of matrices whose size is specified by the number of overspecifying parameters. Here, we report the run times for the models at different array lengths. All experiments were performed on an AMD Ryzen 7735U (8 cores, 16 threads), and 16 GB RAM, running Windows 11. The implementation was written in MATLAB R2023a. The simulations shown above were executed using parallel processing with 8 workers, parallelizing along the 200 replicates. For the MA(1) simulations, the average run time was 2117.8 seconds, with a minimum of 1746.8 seconds and a maximum of 2479.8 seconds. The Mat\'ern simulations took considerably longer with an average run time of 9143.1 seconds, a minimum of 8439.8 seconds and a maximum of 11049 seconds. This increase in run time is due to the increased complexity of the Mat\'ern covariance structure, as the Mat\'ern requires calls to the Bessel function to determine the covariance values, and only a few entries in the relevant large matrices are identically zero, unlike in the MA(1) covariance structure.

There are several modifications one could make to the algorithm which could potentially provide computational benefits, but which are outside of the scope of this work. First, in the current implementation, the step size of the MCMC is constant. Implementing a Hamiltonian MCMC version could allow a practitioner to reduce the number of steps of the chain, though this would have to be balanced with the cost of computing the step sizes. Further, parallelization could provide a significant reduction in time, as the algorithm proposed is embarrassingly parallel over the number of signature matrices sampled at each step. In testing the benefits of this parallelization, we see that for a small number of signature matrices (recall that we use only 8 for all the simulations above), any computational benefit is negated by the cost of scheduling the parallelization, and thus this may only prove beneficial for larger array sizes, where a larger number of signature matrices could be advisable. 

One practical adjustment to the method could be to use an approach similar to composite likelihood. We note that while the algorithm scales poorly in the size of the arrays, in the number of arrays, i.e. independent observations, the run time scales quite well, as this affects only the evaluation of the likelihood. Where appropriate, as for stationary time series data, one could split a sample as in composite likelihood: splitting the sample into many smaller observations as opposed to a full likelihood. We test this method by once again simulating MA(1) data, this time as four independent observations of length 100. We again run our fiducial MCMC algorithm on the original $4 \times 100$ data, but also split each observation of length 100 into 81 observations of length 20, leading to a $241 \times 20$ matrix. Figure \ref{compositeLikelihood} visualizes this transformation of the data in a smaller case. We initiate the chain at the true value of $\theta = (.5,6)$, and run our algorithm for 6000 total iterations, 1000 being burn-in, on both of these data sets. We can immediately see the computational benefit, as the full likelihood approach took 12200.71 seconds (over three and a quarter hours), while the composite likelihood approach took 209.37 seconds (just over three minutes).
\begin{figure}
\begin{equation*}
    \begin{bmatrix}
        a_1 & b_1 & c_1 & d_1 & e_1 & f_1 & g_1\\
        a_2 & b_2 & c_2 & d_2 & e_2 & f_2 & g_2
    \end{bmatrix} \implies 
    \begin{bmatrix}
        a_1 & b_1 & c_1 & d_1 & e_1\\
        b_1 & c_1 & d_1 & e_1 & f_1\\
        c_1 & d_1 & e_1 & f_1 & g_1\\
        a_2 & b_2 & c_2 & d_2 & e_2\\
        b_2 & c_2 & d_2 & e_2 & f_2\\
        c_2 & d_2 & e_2 & f_2 & g_2\\
    \end{bmatrix}
\end{equation*}
\caption{Visualization of the modification of the data for the composite likelihood approach. Here a $2 \times 7$ matrix becomes a $6 \times 5$ matrix.}
\label{compositeLikelihood}
\end{figure}

Beyond computational efficiency, we see that in many ways the composite likelihood approach performs better than the full likelihood. The means of the MCMC chains were similar and both fairly accurate: $(.5804, 5.5654)$ for the composite likelihood and $(.5551, 5.8342)$ for that of the full likelihood. However, we see that the mixing was affected by array size. Figure \ref{fig:CompositeTraceplots} displays the trace plots for both datasets. We note that the composite data resulted in a 43.58\% acceptance rate, while the full likelihood resulted in only 15.02\%. This is likely due to the number of signature matrices sampled at each step, as it was kept constant at 8 for both datasets. Since the total possible number of signature matrices is $2^{d-1}$, where $d$ is the array length, the poor mixing is likely due to not sampling enough signature matrices at the larger array size. In principle, one can address this by increasing the number of signature matrices sampled at each iteration, though this will only increase the computational load.

\begin{figure}
\centering
\begin{minipage}{.5\textwidth}
\centering
    \includegraphics[width=\textwidth]{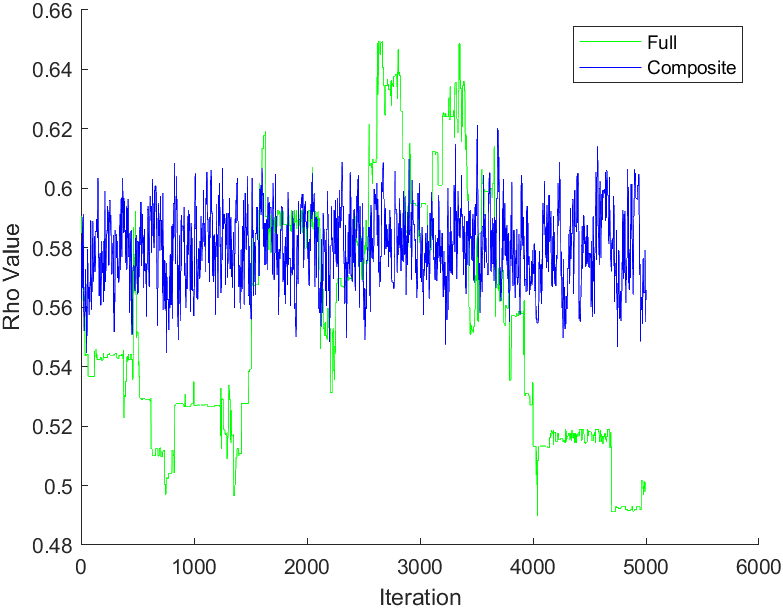}
\end{minipage}\hfill
\begin{minipage}{.5\textwidth}
\centering
    \includegraphics[width=\textwidth]{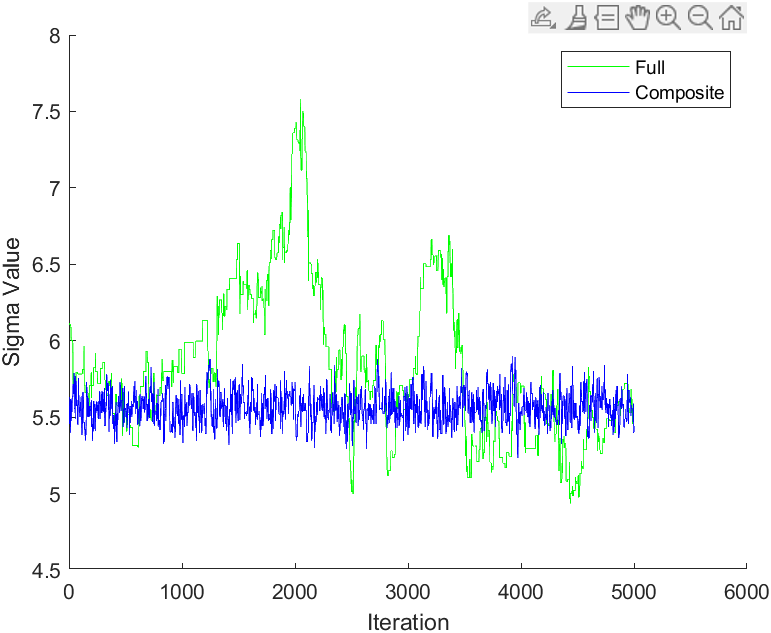}
\end{minipage}
\caption{Trace plots for both parameters of the MA(1) for the full likelihood and composite likelihood approaches. Note that the mixing of the composite likelihood is good, while that of the full likelihood is subpar.}
\label{fig:CompositeTraceplots}
\end{figure}

\begin{figure}
\centering
\begin{minipage}{.33\textwidth}
\centering
     \includegraphics[width=\textwidth]{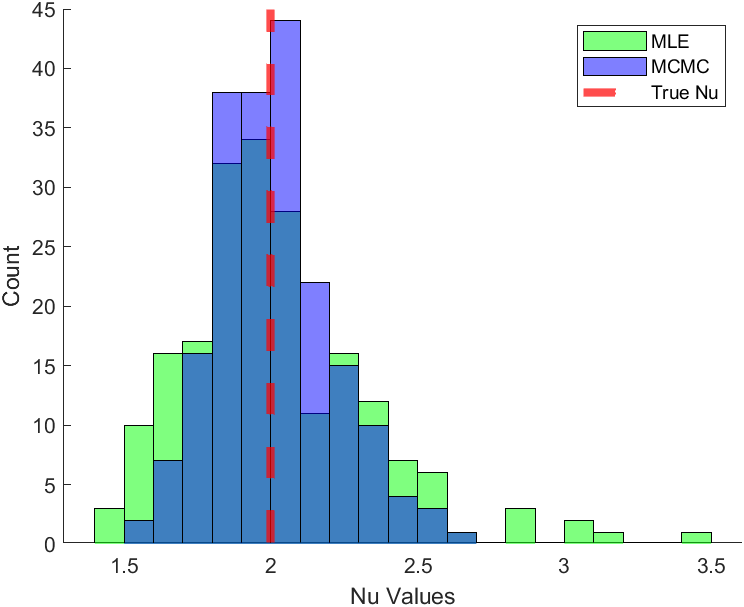}
\end{minipage}\hfill
\begin{minipage}{.33\textwidth}
    \centering
    \includegraphics[width=\textwidth]{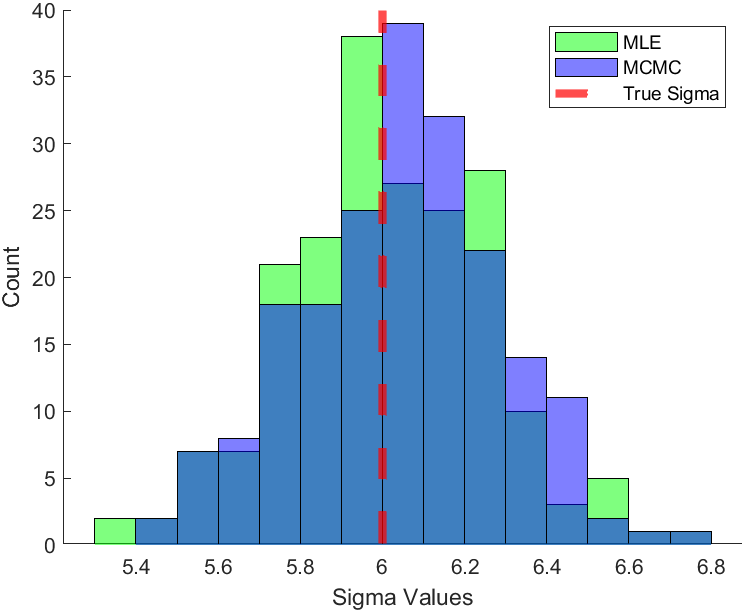}
\end{minipage}
\begin{minipage}{.33\textwidth}
    \centering
    \includegraphics[width=\textwidth]{EstimatedRhos.png}
\end{minipage}
\caption{Estimates of Mat\'ern parameters from 200 runs of our algorithm, compared to the MLE for the same data. From left to right, the estimated parameters are $\nu, \sigma^2,\rho$, with true values of $2,6$, and $1$, respectively.}
\label{fig:Matern sims}
\end{figure}

\begin{figure}
\centering
\begin{minipage}{.33\textwidth}
\centering
     \includegraphics[width=\textwidth]{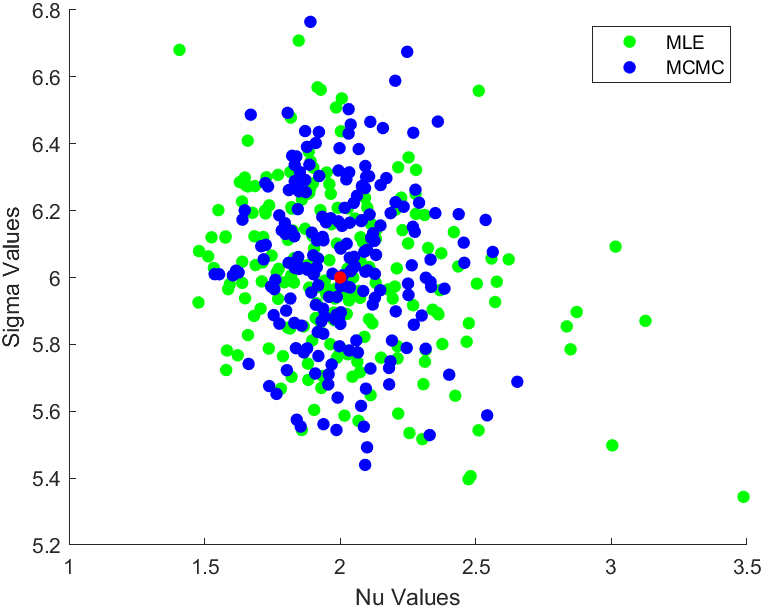}
\end{minipage}\hfill
\begin{minipage}{.33\textwidth}
    \centering
    \includegraphics[width=\textwidth]{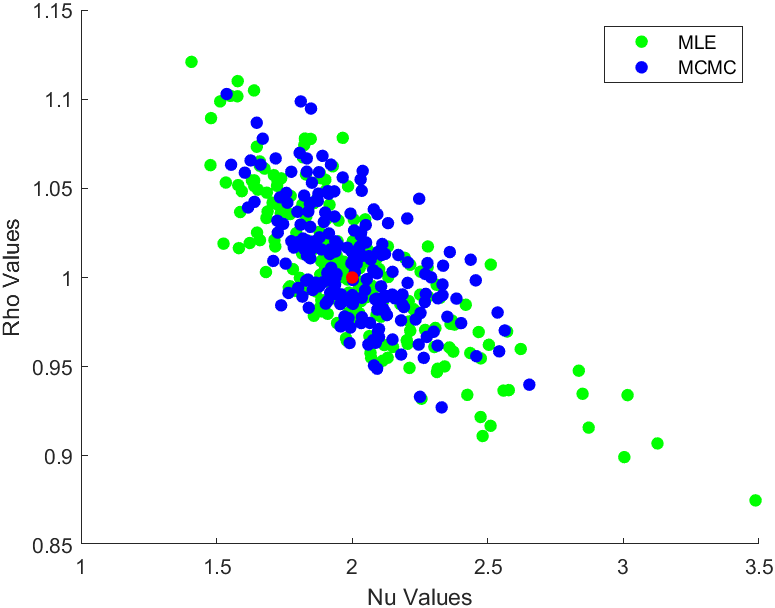}
\end{minipage}
\begin{minipage}{.33\textwidth}
    \centering
    \includegraphics[width=\textwidth]{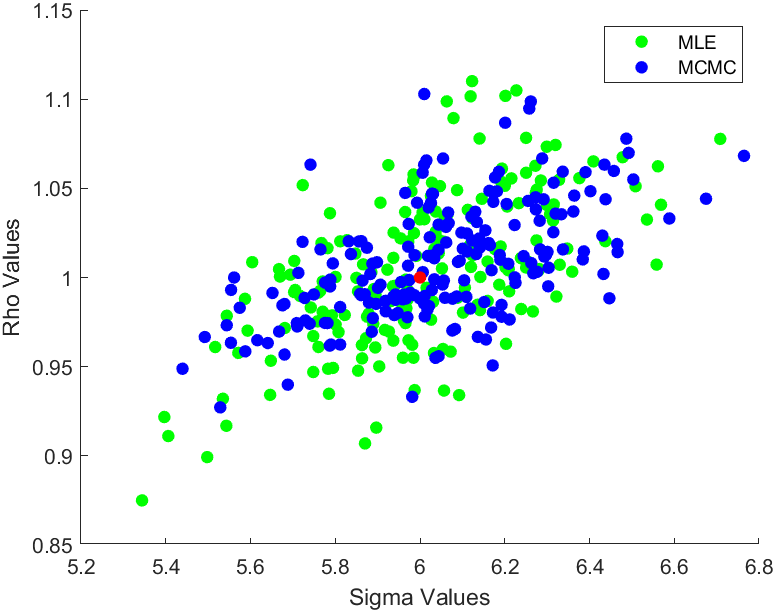}
\end{minipage}
\caption{Scatter plot of estimated parameter values for each of the 200 runs on simulated Mat\'ern data. The red dots indicate the true parameter values.}
\label{fig:Matern scatter}
\end{figure}

\begin{table}
\centering
\begin{tabular}{l|l|l|l|l|l|}
\cline{2-6}
Model  & Parameter & \begin{tabular}[c]{@{}l@{}}True\\ Value\end{tabular} & \begin{tabular}[c]{@{}l@{}}Average\\ Estimated Value\end{tabular} & \begin{tabular}[c]{@{}l@{}}2.5th\\ Quantile\end{tabular} & \begin{tabular}[c]{@{}l@{}}97.5th\\ Quantile\end{tabular} \\ \hline
MA(1)    & $\sigma^2$    & 6               & 6.0161   & 5.2563 & 6.8623    \\ \cline{2-6} 
       & $\rho$       & 0.5                & 0.5035   & 0.4396 & 0.5697     \\ \hline
Mat\'ern & $\nu$        & 2                & 2.0065   & 1.8051 & 2.2127    \\ \cline{2-6} 
       & $\sigma^2$    & 6               & 6.0477   & 5.6291 & 6.5004    \\ \cline{2-6} 
       & $\rho$       & 1                & 1.0084   & .9734 & 1.0429    \\ \cline{2-6} 
\end{tabular}
\caption{Table of average parameter estimates and quantiles for each parameter, for both models.}
\label{estimated vals}
\end{table}
\section{Proof} \label{Fiducial proofs}
In this section we provide proofs regarding our MCMC algorithm. Namely, we show that the stationary distribution of the MCMC algorithm is the pseudo-marginal distribution on $\theta$, which will equal the GFD on $\theta$, once the effect of the $\{Z_i\}$ is averaged out. Additionally, we call on previous results to argue a Bernstein-von Mises-like result for fiducial distributions, which argues that the fiducial distribution is asymptotically equivalent in total variation distance to a normal distribution, centered at a sufficient statistic of the parameter vector $\theta$.
\begin{proposition}
    Assume the stationary distribution
    \begin{equation*}
         \pi(\theta, \{Z_i\})=\frac{1}{m}f(x,\Sigma(\theta))\sum_{i=1}^kJ(A_i,\Lambda,\theta)1_{\{\theta,Z_i\}}\left(k2^{(k-1)(d-1)}\right)^{-1}.
    \end{equation*}
    The proposed algorithm satisfies detailed balance.
\end{proposition}Before we being the proof, recall the definition of detailed balance. If a Markov chain defined on $\Theta$ has transition probability $p(\theta,\theta')$ and stationary distribution $\pi(\theta)$ for $\theta,\theta' \in\Theta$. Detailed balance is satisfied if
\begin{equation*}
    \pi(\theta)p(\theta,\theta') = \pi(\theta')p(\theta',\theta).
\end{equation*}
\noindent \textit{Proof.}
Let $\left(\theta,\{Z_i\}\right)$ be the current position of the algorithm and $\left(\theta',\{Z_i'\}\right)$ be the proposal. Let $k$ be the number of signature matrices sampled at each step of the algorithm, $q$ be the number of signature matrices kept between transitions, $w$ be the number of signature matrices shared by $\{Z_i\}$ and $\{Z_i'\}$. Also let $A_i$ be the skew-symmetric matrix of the Cayley transform corresponding to the signature matrix $Z_i$. The density of the transition from $\theta$ to $\theta'$ is given by the normal distribution. To transition from $\{Z_i\}$ to $\{Z_i'\}$ (conditioned on $\theta$ and $\theta'$), then we must select $q$ of the $w$ shared matrices to be kept between transitions, and then select the other $k-q$ at random from the $2^{d-1}$ possible signature matrices. Then, the transition probability density from $\left(\theta,\{Z_i\}\right)$ to $\left(\theta',\{Z_i'\}\right)$ is given by
\begin{align*}
    &p\left(\theta,\{Z_i\},\theta',\{Z_i'\}\right) =\\
    &\phi(\theta,\theta')\binom{w}{q}\binom{k} {q}^{-1}\left(\frac{1}{2^{d-1}}\right)^{k-q}\min\left[1,\frac{f(x,\Sigma(\theta'))}{f(x,\Sigma(\theta))}\frac{\sum_{i=1}^{k}J(A_i',\Lambda,\theta')1_{\{\theta',Z_i'\}}}{\sum_{i=1}^{k}J(A_i,\Lambda,\theta)1_{\{\theta,Z_i\}}}\right],
\end{align*}if $w\geq q$ and $\{Z_i'\}$ contains at least one signature matrix such that $1_{\{\theta',Z_i'\}}=1$, and is zero otherwise. Here, $\phi(\theta, \theta')$ represents the normal density centered at $\theta$ and evaluated at $\theta'$, and the combinatorial terms represent the probability of selecting the set $\{Z_i'\}$. Let $m$ be the normalizing constant for the generalized constrained fiducial distribution. The stationary distribution we are targeting is given by 
\begin{equation*}
    \pi(\theta, \{Z_i\})=\frac{1}{m}f(x,\Sigma(\theta))\sum_{i=1}^kJ(A_i,\Lambda,\theta)1_{\{\theta,Z_i\}}\left(k2^{(k-1)(d-1)}\right)^{-1},
\end{equation*}
where $m$ is the normalizing constant that makes this a valid probability distribution. Here, $\left(k2^{(k-1)(d-1)}\right)^{-1}$ is a normalizing constant related to the $Z_i$ whose form will be clear in later calculations. To analyze detailed balance, we write

\begin{align*}
    &\pi(\theta,\{Z_i\}) p(\theta,\{Z_i\},\theta',\{Z_i'\})\\
    &= \frac{1}{m}\min\left[f(x,\Sigma(\theta))\sum_{i=1}^{k}J(A_i,\Lambda,\theta)1_{\{\theta,Z_i\}},\left.f(x,\Sigma(\theta'))\sum_{i=1}^{k}J(A_i,\Lambda,\theta')1_{\{\theta',Z_i'\}}\right]\right.\\
    &\hspace{10mm}\cdot\left(k2^{(k-1)(d-1)}\right)^{-1}\phi(\theta,\theta')\binom{w}{q}\binom{k}{q}^{-1}\left(\frac{1}{2^{d-1}}\right)^{k-q}\\
    &= \pi(\theta',\{Z_i'\})p(\theta',\{Z_i'\},\theta,\{Z_i\})\\
\end{align*} Thus, detailed balance is satisfied. We now show that this pseudo-marginal distribution reduces to the GCFD when the effect of the $Z_i$ is integrated out. The change of variable formula gives us the relationship
\begin{equation*}
    J(\theta) = \sum_{i=1}^{2^{d-1}}J(A_i,\Lambda,\theta)1_{\{\theta,Z_i\}}.
\end{equation*}
Let $\sigma_i$ be the $i$th (of an arbitrary ordering) of the $2^{k(d-1)}$ possible combinations of the sets of $Z_i$. That is, $\sigma_i: \{1,..,k\} \rightarrow \{1,...,2^d\}^k$ is such that $\sigma_i(j)$ gives $Z_j$ in the $i$th combination. Then, we have

\begin{align*}
    &\sum_{i=1}^{2^{k(d-1)}}\frac{1}{m}f(x,\Sigma(\theta))\sum_{j=1}^kJ(A_{\sigma_i(j)},\Lambda,\theta)1_{\{\theta,Z_{\sigma_i(j)}\}}\left(k2^{(k-1)(d-1)}\right)^{-1}\\
    &= \frac{1}{m}f(x,\Sigma(\theta))\left(k2^{(k-1)(d-1)}\right)^{-1}\sum_{i=1}^{2^{k(d-1)}}\sum_{j=1}^kJ(A_{\sigma_i(j)},\Lambda,\theta)1_{\{\theta,Z_{\sigma_i(j)}\}}.
\end{align*}
There are $2^{k(d-1)}$ total combinations of $k$ matrices. By the symmetry of this double sum about $Z_i$, each $Z_i$ appears in the sum $k2^{(k-1)(d-1)}$ times. We write
\begin{align*}
    &= \frac{1}{m}f(x,\Sigma(\theta))\left(k2^{(k-1)(d-1)}\right)^{-1}k2^{(k-1)(d-1)}\sum_{i=1}^{2^{d-1}}J(A_i,\Lambda,\theta)1_{\{\theta,Z_i\}}\\
    &= \frac{1}{m}f(x,\Sigma(\theta))J(\theta),
\end{align*}
which is exactly the GCFD.
\qed

Because the distribution $\pi$ satisfies detailed balance, it is an invariant distribution of the Markov Chain. To show that this distribution is unique, we need to argue that the chain is aperiodic, and $\pi$-irreducible. Aperiodicity is clear, as the chain has a non-zero probability of remaining at the same point at each step. Recall that the definition of $\pi$-irreducible is that any set given positive probability by the stationary distribution $\pi$ is accessible in a finite number of steps from any point in the state-space. To argue that the Markov Chain is $\pi$-irreducible, we must first specify a starting point $(\theta_1, \{Z_i\})$, and a target set with positive mass, which we will decompose into a set $\Theta' \subset\Theta$ and $\mathcal{Z}$, a set of sets of signature matrices of size $k$. If the target has positive mass, then $\Theta'$ must have a non-empty interior, and at least one set, $\{Z_i'\}$ which has at least one permissible signature matrix somewhere on the interior of $\Theta'$. Without loss of generality, let $\{Z_i'\} \in \mathcal{Z}$ be such a set, and let $Z_1'$ be permissible somewhere on the interior of $\Theta'$. Consider the transition from $(\theta_1,\{Z_i\})$ to a point in $(\Theta_1 \cap \Theta')\times\{Z_1',...,Z_{k-q}',Z_{k-q+1},...,Z_{k}\}$. Then, we know that this is a valid transition, since any point in $\Theta$ is accessible from any other in a single step, and there exists at least one signature matrix, namely $Z_1'$ which is permissible. Then, the chain can transition, staying within $(\Theta_1 \cap \Theta')$ so that the transition is valid, and updating the $k-q$ signature matrices which change at each step until the step where the set of signature matrices is equal to $\{Z_i'\}$. This whole process can always be done in at most $\lceil\frac{k}{k-q}\rceil$ steps. Thus, the chain is also $\pi$-irreducible, and $\pi$ is the unique stationary distribution. 

We have that the Markov Chain has the pseudo-marginal distribution as its steady state solution. To show that this leads to a consistent estimator, we cite Theorem 3.1 of \citet{borgert2024}, which states that under regularity conditions, as the sample size $n\to \infty$
\begin{equation*}
    ||r(\theta|y)_{H_n|Y_n}- N(\Delta_{n,\theta_0},I^{-1}_{\theta_0})||_{TV}\to0.
\end{equation*}Here, $r(\theta|y)_{H_n|Y_n}$ represents the localized fiducial distribution constrained to a ball of radius dominated by $\sqrt{n}$, and $N(\Delta_{n,\theta_0},I^{-1}_{\theta_0})$ represents the normal distribution, with covariance matrix equal to the inverse of the Fisher Information and centered at a sufficient statistic $\Delta_{n,\theta_0}$. This result is essentially a Bernstein-von Mises result for the fiducial distribution. In other words, as the sample size grows the ``posterior-like" fiducial distribution will asymptotically approximate the normal distribution centered at the sufficient statistic. This ensures that the samples of our MCMC algorithm will asymptotically approximate such a sufficient statistic.

\citet{borgert2024} give the assumptions necessary for the Bernstein-von Mises result:
\begin{otherassumption}\label{A1}
    The experiment $\{P_\theta:\theta\in\Theta\}$ is differentiable in quadratic mean at $\theta_0$ with non-singular Fisher information matrix $I_\theta$.
\end{otherassumption}
\begin{otherassumption}
    There exists a $\sigma$-finite measure absolutely continuous with respect to the Lebesgue measure on $||\theta-\theta_0||\leq D$ for some $D>0$, with density $\pi$ such that
    \begin{equation*}
        \sup_{\theta:||\theta-\theta_0||\leq D}\frac{|J(y,\theta)-\pi(\theta)|}{J(y,\theta)}\xrightarrow{P_{\theta_0}}0.
    \end{equation*}
\end{otherassumption}
\begin{otherassumption}
    The limiting density $\pi(\theta)$ is continuous and positive at $\theta_0$.
\end{otherassumption}
\begin{otherassumption}
    There exist finite measures $P_\theta^1$, $P_\theta^2$ absolutely continuous with respect to the Lebesgue measure having densities $p^1(y,\theta),p^2(y,\theta)$, respectively, such that
    \begin{equation*}
        p(y,\theta)=p^{1}(y,\theta)p^{2}(y,\theta).
    \end{equation*}
    Moreover, there exists a finite measure on $\theta$, absolutely continuous with respect to the Lebesgue measure, with density $\gamma(\theta)$ so that for the same $D>0$ as in Assumption 2,
    \begin{equation*}
        P_{\theta_0}(p^2(y,\theta)J(y,\theta)I_{\{||\theta-\theta_0||>D\}}\leq\gamma(\theta))\to1.
    \end{equation*}
\end{otherassumption}
\begin{otherassumption}
    For every $\epsilon>0$, there exists a sequence of tests $\delta_n$ and a constant $c>0$ such that for $||\theta-\theta_0||\geq\epsilon$,
    \begin{equation*}
        P_{\theta_0}(\delta_n)\to0\hspace{3mm}\text{and}\hspace{3mm} P_\theta(1-\delta_n)\leq e^{-cn(||\theta-\theta_0||^2\wedge1}\hspace{3mm}\text{and}\hspace{3mm} P^{1}_\theta(1-\delta_n)\leq e^{-cn||\theta-\theta_0||^2\wedge1}.
    \end{equation*}
\end{otherassumption} We note that Assumption \ref{A1} is satisfied by our assumptions at the beginning of Section \ref{Fiducial derivations}, though these assumptions have not been verified for either of the models which we explore in this work. These assumptions are quite similar to those typically required for the usual Bernstein-von Mises result, for example in \citet{vaart_1998}. Some examples of models which satisfy these conditions, including a free knot spline model which follows a Gaussian distribution, are given in \citet{borgert2024}.\\

\noindent\textbf{Conflicts of Interest}

\noindent The authors declare no conflicts of interest.\\

\noindent\textbf{Data Availability Statement}

\noindent The data that support the findings of this study are openly available in Constrained-Fiducial-Inference-for-Gaussian-Models at \url{https://github.com/fluryh/Constrained-Fiducial-Inference-for-Gaussian-Models}.

\printbibliography

\appendix
\end{document}